\documentclass[10pt,conference]{IEEEtran}
\usepackage{lipsum,multicol}
\usepackage{color}

\newtheorem{theorem}{Theorem}[section]
\newtheorem{lemma}[theorem]{Lemma}
\newtheorem{proposition}[theorem]{Proposition}

\newtheorem{remark}{Remark}
\newcounter{mytempeqncnt}

\usepackage[pdftex]{graphicx}
\usepackage{graphics} %
\usepackage{cite}
\usepackage{color}
\usepackage{float}
\usepackage{algorithmic}
\usepackage{amsmath}
\usepackage{amsfonts}
\usepackage{amssymb}
\usepackage{amsopn}
\usepackage{xspace}
\usepackage{array}
\usepackage{epsfig}
\usepackage{color}
\usepackage{subfigure}
\usepackage{multirow}

\newcommand{\Ex}{\mathop{\bf E\/}}

\newcommand{\rate}{R}

\makeatletter  %% this is crucial
 \renewcommand\subsubsection{\@startsection{subsubsection}{3}{\z@}%
                        {-18\p@ \@plus -4\p@ \@minus -4\p@}%
                        {8\p@ \@plus 4\p@ \@minus 4\p@}%     <-- this is copied from the subsection command
                        {\normalfont\normalsize\bfseries\boldmath
                         \rightskip=\z@ \@plus 8em
                          \pretolerance=10000 }}
                          \makeatother   %% this is crucial

\begin{document}

\title{Cross-Layer Modeling of Randomly Spread CDMA Using Stochastic Network Calculus}

\author{\IEEEauthorblockN{Kashif Mahmood\IEEEauthorrefmark{1}, Mikko Vehkaper{\"a}\IEEEauthorrefmark{2}, Yuming Jiang\IEEEauthorrefmark{1}}       % <-this % stops a space

\IEEEauthorblockA{\IEEEauthorrefmark{1} Q2S, Norwegian University of Science and Technology,~Norway\\
\IEEEauthorrefmark{2} School of Electrical Engineering and the ACCESS Linnaeus Center,~KTH,~Sweden  }}

\maketitle

\begin{abstract}
Code-division multiple-access (CDMA) has the potential to support traffic sources with a wide range of quality of service (QoS) requirements.
The traffic carrying capacity of CDMA channels under QoS constraints (such as delay guarantee) is, however, less well-understood.
In this work, we propose a method based on stochastic network calculus and large system analysis to quantify the maximum traffic that can be carried by a multiuser CDMA network under the QoS constraints.
At the physical layer, we have linear minimum-mean square error receivers and adaptive modulation and coding, while the channel service process is modeled by using a finite-state Markov chain.
We study the impact of delay requirements, violation probability and the user load on the traffic carrying capacity under different signal strengths.
A key insight provided by the numerical results is as to how much one has to back-off from capacity under the different delay requirements.
%Furthermore we study the impact of two linear receivers,SUMF and LMMSE, on flow level QoS performance.
%Quantifying the maximum packet arrival rate per user that the network can support such that the QoS requirement is met is a crucial factor for QoS provisioning.
%The maximum traffic arrival rate at the network for a given delay guarantee (\emph{delay constrained throughput}) has been well studied for wired channels.
%However, few results are available for wireless channels.
%In this work, we analyze the delay constrained throughput of a multiuser single antennas CDMA system.
\end{abstract}
%
%
%\begin{IEEEkeywords}
%Delay constrained throughput, DS-CDMA, linear minimum-mean square error (LMMSE) receiver, asymptotic analysis, finite-state Markov channel (FSMC) modeling, stochastic network calculus, moment generating function (MGF).
%\end{IEEEkeywords}

\section{Introduction}
The rapid growth of delay sensitive applications, such as VOIP and IP video, and the scarcity of the radio spectrum require the design of spectrally-efficient systems with quality of service (QoS) support.
Direct-sequence code-division multiple-access (DS-CDMA) \cite{Verdu-1998}
is an efficient technique that supports
flexible scaling of the user population,
%while at the same time offering support to
bursty traffic sources, and a wide range of QoS requirements,
such as delay guarantees.
%CDMA: ease of cellular planning, superior performance in multi-path environments and wide range of operating environments.
%Provision of delay guarantee is an essential QoS requirement.
A fundamental question is: What is the traffic carrying capacity (throughput) of CDMA networks under the given delay constraints?
Answering this question, however, is a challenging task because the existing physical layer channel models do not explicitly take into account the source burstiness and the QoS requirements.
Furthermore, the time varying nature of the wireless channel may itself
cause severe QoS violations.
The above reasons call for a cross-layer approach that can model the link-layer user requirements (e.g. delay guarantee) and at the same time take into account the physical-layer techniques, such as error control coding \cite{Richardson-Urbanke-2008} and multiuser detection \cite{Verdu-1998}.
%
%time varying fading channel and multiuser receivers.
Indeed, predicting the throughput%
\footnote{Throughout the paper, the term \emph{delay} refers to \emph{queuing delay} and the \emph{throughput} as the \emph{delay constrained throughput}.}
regions for a given delay guarantee is seen as one of the grand challenges in multiuser networks~\cite{NIT:Jan08:RethinkingIT}.

%We call this %capacity for a give delay guarantee as the
%the \emph{delay constrained throughput}
%\footnote{Throughout the paper, the term \emph{delay} refers to \emph{queuing delay} and the \emph{throughput} as the \emph{delay constrained throughput}.} of the network.
%
%We define the delay constrained throughput as the \emph{maximum packet arrival rate} per user that the network can support such that the QoS requirement is met.
%%\footnote{We assume that the packets exceeding the delay requirement are dropped at the egress node}.
%The service model is based on the finite state Markov Chain (FSMC) model.
%FSMC model is considered a sufficiently accurate model for representing Rayleigh fading channels.
%
%
%%\subsection{Related Work}
%
%A significant amount of work has been done on CDMA networks e.g.~ \cite{CDMA:LinearMultiuserReceivers:Tse-Hanly-99,CDMA:SpectralEffecRandomSpreading:verdu:99,CDMA:OutputMAI:Zhang-Chong-Tse-01} to quantify the throughput using the classical information theory which concentrates solely on physical layer.
%This work does not taken into account the burstiness of the source and the delay guarantees which are link-layer parameters.
%
%It is clear from the above discussion that a cross-layer approach is required in order to capture the delay and burstiness constraints.
Two front runner theories for cross-layer modeling of wireless networks are the
\emph{effective capacity}~\cite{EBWireless:DWU03}, and the \emph{stochastic network calculus}~\cite{NetCal:Fidler06:MGFfadingChannel,NetCal:Jiang05:StochServGuarantServerModel}.
The former has been used, e.g., for cross-layer modeling of single user multiple-input multiple-output (MIMO) channels with memory~\cite{EBWireless:Tang2007:CrossLayerModelingforQoS}.
An attempt was also made in \cite{CDMA:EffectiveBW:YuF06} to use efficient bandwidth for cross-layer modeling of wireless CDMA networks with linear minimum-mean square error (LMMSE) receivers.
In~\cite{CDMA:EffectiveBW:YuF06}, however, only
\emph{memoryless channels} were considered
and the applicability of the proposed approach for quantifying the traffic carrying capacity of CDMA networks under delay constraints is not clear.
\emph{Stochastic network calculus} (NetCal)~\cite{NetCal:Jiang08:SNC:Bk}, on the other hand, is a more general theory that has been recently applied for QoS analysis of wireless channels with memory~\cite{NetCal:Jiang05:StochServGuarantServerModel,NetCal:Fidler06:MGFfadingChannel}.

In this paper, we propose a cross-layer approach to predict the \emph{traffic carrying capacity} of multiuser CDMA networks under a given delay guarantee.
We make use of moment generating function (MGF) based stochastic NetCal that was proposed in \cite{NetCal:Chang00:PerGuaran:Bk}, and further established and developed in \cite{NetCal:Fidler2006:EoEProbabNetCalWithMGF}.
In stochastic NetCal, a communication system is described by the arrivals at the ingress, the service it provides, and the departures at the egress.
It enables easily the consideration of statistical independence between arrivals and system service.
%A comprehensive view of stochastic NetCal can be found in \cite{NetCal:Jiang08:SNC:Bk}, \cite{fidler_survy_2010}.
%
In contrast to the previous work, for example, in~\cite{CDMA:EffectiveBW:YuF06}, here the multiuser CDMA channel is considered to have \emph{memory}.
A significant amount of work has been done on link-layer channel modeling of such wireless channels (see~\cite{WirelessMarkov:Sadeghi08:FiniteStateMarkov:Survey} and the references therein), especially in the single-user setting.
Most of this work builds on the finite-state Markov channel (FSMC) model,
%, which is a well established higher-layer model for QoS analysis of time varying wireless channel.
%In this model, the received signal-to-noise ratio (SNR) is partitioned into a finite number of states.
which has been successfully applied on the analysis of delay constrained
single antenna \cite{EBWireless:Li07} and MIMO \cite{EBWireless:Tang2007:CrossLayerModelingforQoS} systems, to model the memory in the wireless channels.
To the best of our knowledge, however, none of the works based on FSMC considers CDMA networks or uses
%The cross layer modeling attempt for CDMA networks in ~\cite{CDMA:EffectiveBW:YuF06} does not consider channels with memory.
%To the best of our knowledge,
the stochastic NetCal to analyze the delay constrained throughput of multiuser
networks.

%The concept of effective bandwidth is used in \cite{EBWireless:Li07} to estimate the capacity of the wireless channels with memory using FSMC that varies randomly with time and space.

The goals and major contributions of this work are as follows.
We present a methodology to calculate the \emph{network delay constrained throughput} of a multiuser DS-CDMA system that employs adaptive modulation
and coding (AMC) at the transmitter, and LMMSE multiuser detection at the
receiver.
The wireless channel is modeled via a FSMC and MGF based NetCal is used to compute the throughput of the system.
The formulation is valid for any traffic arrival process 
for which the MGF or a bound on the MGF exists.
In the analysis, we incorporate statistical independence of arriving traffic and the service provided by the channel.
We show the impact of the delay guarantee, violation probability and the user load on the throughput under various conditions, such as signal strength and fading speed.
%In a nutshell, this paper studies the interplay between physical layer interference suppression algorithms (multiuser detectors), the QoS constraints (delay guarantee and delay violation probability) and the network layer throughput.

%The rest of the paper is organized as follows:
%We first present the preliminaries of DS-CDMA system model and LMMSE receivers in Sec.~\ref{sec:prelim}.
%A FSMC modeling of the wireless system is carried out in Sec.~\ref{Sec:DelayConstrainedThAnalysis} for calculating the delay constrained throughput.
%Finally, we calculate the numerical results in Sec.~\ref{Sec:NumericalResults}

\section{The System}
\label{sec:prelim}

The present paper considers a synchronous uplink CDMA channel 
with $K$ users.
Each mobile station (MS) is equipped with a single transmit antenna and the
base station (BS) has a single receive antenna.  %In the following, we use the term user and MS interchangeably.
Perfect channel state information (CSI) is assumed to be available at the BS, while the MSs have no direct access to CSI.
%but they receive information for AMC
%through an error free feedback. %low rate error free feedback channel.
%because it is difficult and often not practical to assume CSI at each mobile station (MS).
%We use $[\cdot]^\dag$ to denote the matrix conjugate transpose,
%
At the physical layer, a discrete time signal model indexed by
$n = 1,2,\ldots$ is considered.
We let the number of \emph{active users} $K_n \leq K$ at any discrete
time instant $n$ be a random variable, modeling for example
the bursty information sources.
%Due to bursty information sources,
%the number of \emph{active users} $K_n \leq K$ at some discrete
%time instant $n$ is a random variable.
%Without loss of generality, we assume in the following that
%the users $k = 1,\ldots,K_{n}\leq K$ are active for any given
%time instant.
It is assumed that the BS
knows the set of active users at all times.

With the above assumptions,
the discrete time received signal after matched filtering and sampling
reads \cite{Verdu-1998,CDMA:LinearMultiuserReceivers:Tse-Hanly-99}
\begin{equation}
  \mathbf{y}_{n} = %\frac{1}{\sqrt{N}}
  \sum_{k = 1 }^{K} %\sqrt{P_k}
  h_{k} \mathbf{s}_{k} b_{k,n} + \mathbf{v}_{n}  \in \mathbb{C}^{M},
  \quad n = 1,2,\ldots,
  \label{eq:sys_model}
\end{equation}
%{\color{red} PROBLEM! CW LENGTHS ARE DIFFERENT!}
where
$b_{k,n}$ is the code symbol transmitted by the
$k$th user at time instant $n$.
We let $b_{k,n} = 0$ if the user $k$ is not active and
$b_{k,n} \in \mathbb{C}\setminus \{0\}$ otherwise.
The code symbols of the active users are assumed to be independent with
zero mean and unit variance.
%$T$ denotes the length of the code words
%$\mathbf{b}_{k} = [b_{k,1}, \ldots,b_{k,n}]^{\mathsf{T}}\in\mathbb{C}^{T}$
The length (the spreading factor) of the signature sequences $\{\mathbf{s}_{k}\}_{k=1}^{K}$ is $M$ for all users. The random vector $\mathbf{v}\sim\mathsf{CN}(\mathbf{0}, \sigma^{2}\mathbf{I}_{M})$
represents samples of additive white Gaussian noise (AWGN), where
$\mathbf{I}_{M}$ denotes
the identity matrix of size $M \times M$, and $\mathsf{CN}(\boldsymbol{\mu},
\mathbf{\Psi})$ stands for the circularly symmetric complex Gaussian
(CSCG) distribution
with mean~$\boldsymbol{\mu}$ and covariance~$\mathbf{\Psi}$.
The channels between the MSs
and the BS are assumed to be block fading
(see, e.g., \cite{Fabregas-Caire-2006})
%\cite{Biglieri-Proakis-Shamai-98} 
with a fixed
%block length
%$T_{\mathrm{b}}$ which can be less than or equal to the
coherence time.
Throughout the paper we assume that
%$T\leq T_{\mathrm{b}}$
the code words transmitted by the users are shorter than
the coherence time, and omit the block index for
notational simplicity.
The fading coefficients $h_{k} \sim\mathsf{CN}(0, 1)$
for all $k=1,\ldots,K$ are considered to be independent and
identically distributed (IID), but
extension to unequal average received powers is straightforward.
%with equal average powers%
%\footnote{This can be generalized to un-equal average powers.}
%$P_{k}=1$ for all $k=1,\ldots,K$.
Random spreading \cite{Verdu-1998,CDMA:LinearMultiuserReceivers:Tse-Hanly-99,
%CDMA:SpectralEffecRandomSpreading:verdu:99,
CDMA_ImpactofFreqFlatFading:Shamai-Verdu:01} is assumed to be
employed at the MSs, so
that the signature sequences $\{\mathbf{s}_{k}\}_{k=1}^{K}$ are
IID CSCG random vectors
$\mathbf{s}_{k}\sim\mathsf{CN}(\mathbf{0}, \frac{1}{M} \mathbf{I}_{M})$.
 For future convenience,
we define the user load of the system at time instant $n$ as the ratio
$\alpha_{n} = K_{n} / M$, and the average received signal to noise ratio (SNR) as
%\begin{equation}
%  \overline{\mathsf{SNR}} = \frac{1}{\sigma^{2}}.
%\end{equation}
%\begin{equation}
  $\mathsf{SNR}_{\mathrm{avg}} = 1 / \sigma^{2}.$
%\end{equation}
%The channel gains are normalized so that $\Ex [|h_{m,n}|^2]=1,\forall m,n$.

%%%%
%
%
\subsection{LMMSE Receiver for Multiuser CDMA}
\label{ssec:lmmse-mud}

Multiuser receivers mitigate, in addition to background noise, the interference between the users.
It is well understood that
\emph{the choice of multiuser detector has a significant effect on the system capacity} (see, e.g., \cite{Verdu-1998,CDMA:LinearMultiuserReceivers:Tse-Hanly-99, %CDMA:SpectralEffecRandomSpreading:verdu:99,
CDMA_ImpactofFreqFlatFading:Shamai-Verdu:01}). The effect of
\emph{delay constraints}
on the throughput of the CDMA network is, however, not yet well understood.
Here we consider the impact of such effects on the throughput obtained by
a multiuser CDMA system equipped with the LMMSE receiver~\cite{Verdu-1998}.
%that is known to minimize the MSE among all linear receivers.
Extensions to other receivers are also possible, but
they have been omitted from the present paper due to
space constraints.

Let us consider without loss of generality the detection of
code symbol $b_{1,n}$ of the first user when the mobile
stations $k = 1,\ldots,K_{n}\leq K$ are active.
The LMMSE receiver forms the decision variable for the first
user at time instant $n$ as
%
%
%\begin{equation}
  $z_{1,n} = \mathbf{c}^{\mathsf{H}}_{1,n} \mathbf{y}_{t},$
  %\label{eq:dec_variable}
%\end{equation}
%
%
where $\mathbf{c}_{1,n}\in\mathbb{C}^{N}$ and the
superscript $(\cdot)^{\mathsf{H}}$ denotes for the conjugate
transpose of the matrix.
If the BS knows perfectly the channel coefficients
$\{h_{k}\}_{k=1}^{K}$ and the signature sequences
$\{\mathbf{s}_{k}\}_{k=1}^{K}$ of all users, we can define
\begin{align}
%  \mathbf{S} &= [\mathbf{s}_{1}\;\cdots\;\,\mathbf{s}_{K}]
%  \in \mathbb{C}^{N \times K}, \\
  \mathbf{S}_{1,n} &= [\mathbf{s}_{2}\;\cdots\;\,\mathbf{s}_{K_{n}}]
  \in \mathbb{C}^{M \times (K_{n}-1)}, \\
  \mathbf{D}_{1,n} &= \mathrm{diag}(|h_{2}|^{2}, \ldots,|h_{K_{n}}|^{2})
  \in \mathbb{C}^{(K_{n}-1) \times (K_{n}-1)}, \\
  \mathbf{M}_{1,n} &=  \mathbf{S}_{1,n}  \mathbf{D}_{1,n}  \mathbf{S}_{1,n}^{\mathsf{H}}
  + \sigma^{2} \mathbf{I}_{K_{n}-1}
  \in \mathbb{C}^{(K_{n}-1) \times (K_{n}-1)},
\end{align}
and the LMMSE detector for the first user becomes
\cite{Verdu-1998,CDMA:LinearMultiuserReceivers:Tse-Hanly-99}
\begin{align}
 \mathbf{c}_{1,n}
%  \underset{ \mathbf{c}^{\mathsf{H}}_{1}\in\mathbb{C}^{1 \times N}}{\arg \min} \Ex\left\{
%  \left.\left| b_{1} - \mathbf{c}^{\mathsf{H}}_{1} \mathbf{y} \right|^{2}
%  \,\right|\,
%  \{\mathbf{s}_{k}\}_{k=1}^{K}, \{h_{k}\}_{k=1}^{K}
%  \right\} \\
  =\frac{h_{1}}{1+|h_{1}|^{2}
  \mathbf{s}_{1}^{\mathsf{H}}\mathbf{M}^{-1}_{1,n}\mathbf{s}_{1}}
   \mathbf{M}^{-1}_{1,n} \mathbf{s}_{1}.
   \label{eq:lmmse_c1}
\end{align}
%It is then easy to show that after the LMMSE receiver,
The conditional output SINR for the LMMSE receiver reads
\begin{equation}
  %\mathsf{sinr}_{1}
  \gamma_{1,n}
  \left(\{\mathbf{s}_{k}\}_{k=1}^{K}, \{h_{k}\}_{k=1}^{K}\right)
  =  |h_{1}|^{2}\mathbf{s}_{1}^{\mathsf{H}}\mathbf{M}^{-1}_{1,n}\mathbf{s}_{1}.
  \label{eq:sinr_conditioned_on_h_and_s_lmmse}
\end{equation}
For the purpose of the cross-layer analysis, however, \eqref{eq:sinr_conditioned_on_h_and_s_lmmse} is not a
convenient starting point due to its dependence on the signature
sequences and the channels of all active users.

\subsection{Large System Limit and the Decoupling Principle}
\label{ssec:large_system}

To simplify the physical layer model, we consider throughout 
the paper the \emph{large system limit} in which
both the number of active users $K_{n} \leq K$ and the spreading factor $M$ grow without bound with
a fixed ratio $\alpha_{n} = K_{n} / M$, for all $n=1,2,\ldots$.
In this asymptotic region,
the multiuser CDMA channel \eqref{eq:sys_model} decouples under very general conditions to a set of single-user channels
\cite{CDMA:LinearMultiuserReceivers:Tse-Hanly-99, CDMA:OutputMAI:Zhang-Chong-Tse-01} that do not depend on the signature sequences
$\{\mathbf{s}_{k}\}_{k=1}^{K_{n}}$.  For the case of LMMSE receiver
considered in Section~\ref{ssec:lmmse-mud},
the SINR after multiuser detection simplifies in the large system limit to
\cite{CDMA:LinearMultiuserReceivers:Tse-Hanly-99, CDMA:OutputMAI:Zhang-Chong-Tse-01}
\begin{equation}
  %\mathsf{sinr}^{\mathrm{su}}_{1}
  \gamma_{1,n}
  = p_{1} / \beta_{1,n},  
  %\frac{p_{1}}{\beta_{1,n}},
  \label{eq:instantaneous_sinr_lmmse}
\end{equation}
%
%We are interested in the distribution of
%\eqref{eq:sinr_conditioned_on_h_and_s_lmmse} that does not explicitly
%depend on $\{\mathbf{s}_{k}\}_{k=1}^{K},$ and $\{h_{k}\}_{k=2}^{K}$.
%%Luckily,
%From \cite{CDMA:LinearMultiuserReceivers:Tse-Hanly-99} and \cite{CDMA:OutputMAI:Zhang-Chong-Tse-01}, we
where $p_{1}$ represents the instantaneous channel power
in the decoupled single-user channel and is drawn according to
the probability density function (PDF)
\begin{equation}
 f(p_{1}) =
    %\begin{cases}
      e^{-p_{1}},  \qquad p_{1} \geq 0.
      %, \\
      %0 & \text{otherwise}.
    %\end{cases}
    \label{eq:pdf_of_p1}
\end{equation}
The noise variance $\beta_{1,n}$ in the equivalent single-user system
is the solution to the fixed point
equation
%$\mathsf{sinr_{1}^{lmmse} = $
\begin{equation}
  \beta_{1,n} = \sigma^{2} + \alpha_{n} \int_{0}^{\infty}
  \frac{p \beta_{1,n}}{p+\beta_{1,n}} e^{-p} \mathrm{d}p.
  \label{eq:b1_fp_lmmse}
\end{equation}
Since the interference becomes Gaussian in the
large system limit \cite{CDMA:OutputMAI:Zhang-Chong-Tse-01},
the decision variables $z_{1,n}$ are statistically equivalent to
the received symbols of the single-user (su) channel
\begin{equation}
  z^{\mathrm{su}}_{1,n} = \sqrt{p_{1}} b_{1,n} + n^{\mathrm{su}}_{1,n}, \qquad n = 1,2,\ldots
  %n^{\mathrm{su}}_{1,n} \sim \mathsf{CN}(0,\beta_{1}),
  \label{eq:equivalent_su_lmmse}
\end{equation}
where
$n^{\mathrm{su}}_{1,n} \sim \mathsf{CN}(0,\beta_{1,n}),
\,n=1,2,\ldots,$ are mutually independent samples of Gaussian noise.
For later use, we denote 
$\overline{\gamma}_{1,n}=\Ex_{p_1}\{\gamma_{1,n}\} = 1/\beta_{1,n}$
for the post-detection SINR \eqref{eq:instantaneous_sinr_lmmse}
that is averaged over the channel gain $p_1$.
%\{n^{\mathrm{su}}_{1,n}\}_{t=1}^{T}$ are IID.
%\subsection{SINR Distributions and the Level Crossing Rate}

%
%

\begin{remark}
  The fixed point equation \eqref{eq:b1_fp_lmmse} is easy to solve
  iteratively and is always guaranteed to converge.  Note also
that $\beta_{1,n}$ has time dependence only through $\alpha_{n}$, so
that in the large system limit only the size of the active user set
matters --- not which users are in it. For the
  purpose of cross-layer analysis, we shall therefore
  \emph{concentrate
  on the equivalent single-user channels} defined by
  \eqref{eq:instantaneous_sinr_lmmse}~--~\eqref{eq:equivalent_su_lmmse},
where
%  In the decoupled signal model, both the the channel
%  coefficients
$\{p_{k}\}_{k=1}^{K}$ and
%  the noise samples
$\{n^{\mathrm{su}}_{k,n}\}_{k=1}^{K}$
  are all mutually independent random variables.
  %and produces an accurate solution  in general after 10-20 iterations.
\end{remark}

\subsection{Adaptive Coded Modulation for Multiuser CDMA}

%Let us assume that the channel changes relatively slowly, so that
%the receiver can estimate the post-detection SINR
%\eqref{eq:sinr_conditioned_on_h_and_s_lmmse} for all the users.
%and feed it back to them via low rate channel.

Consider the single-user channel
\eqref{eq:equivalent_su_lmmse}, where the channel
gain $p_{1}$ is fixed.  Let the
code word of the first user be %of length $T_1$ and written as
$\mathbf{b}_{1} = [b_{1,1}, \ldots,
b_{1,N_1}]^{\mathsf{T}}\in\mathbb{C}^{N_{1}}$, and
denote
%over $\mathbf{b}_{1}$ is denoted %given by the vector
the instantaneous received SNR over it
$\boldsymbol{\gamma}_{1}
=[\gamma_{1,n}, \ldots,\gamma_{1,N_1}]^{\mathsf{T}}$.
%where $\gamma_{1,n}$ are defined in \eqref{eq:instantaneous_sinr_lmmse}.
Assuming the elements of
$\mathbf{b}_{1}$ belong to a discrete set $\mathcal{M}$,
%$\mathbf{b}\in\mathcal{M}^{T}$,
the maximum rate that can be reliably transmitted over
\eqref{eq:equivalent_su_lmmse} reads (see, e.g., \cite{Fabregas-Caire-2006})
%this (conditional) AWGN channel
\begin{align}
  &I^{(N_1)}_{\mathcal{M}}(\boldsymbol{\gamma}_{1}) = \log_{2}(|\mathcal{M}|) - \log_2(e)
  + \frac{1}{N_1} \sum_{n=1}^{N_1} g_{\mathcal{M}}(\gamma_{1,n}),
  \label{eq:cm_capacity}
\end{align}
as $N_1 \to \infty$.  We denoted above
\begin{align}
  &g_{\mathcal{M}}(\gamma) = \nonumber\\
	&- \frac{1}{\pi|\mathcal{M}|}
	\sum_{b\in\mathcal{M}}
	\int e^{-|v|^2}
	\log_{2}
		\Bigg[
	\sum_{\tilde{b}\in\mathcal{M}}
	e^{
  - |v + \sqrt{\gamma} (b - \tilde{b})|^{2}
	}
	\Bigg] \mathrm{d}v,
\end{align}
and the integral is with respect to a complex variable $v\in\mathbb{C}$.
To introduce AMC to the system, we assume that:
\begin{enumerate}
\item If the MS $k$ is \emph{not active}, it uses a code word $\mathbf{b}_{k} = \boldsymbol{0}$;
\item For \emph{active users}, the code word is selected from a sequence
$\mathcal{C}_{k,0}, \mathcal{C}_{k,1}, \ldots, \mathcal{C}_{k,L}$ of random
code books%
\footnote{For simplicity of notation, we assume that code
words of any desired lengths can be picked from all code books
$\mathcal{C}_{k,l}$.} 
that have rates
$0 = R_{0} < R_{1} < \cdots < R_{L} < \infty$.
\end{enumerate}
%where $\mathcal{C}_{k,l}\subseteq\mathcal{M}_{l}^{T}$
All modulation sets
$\mathcal{M}_{l}\subset\mathbb{C}$, $l=0,1,\ldots,L$,
satisfy $\Ex |b_{k,n}|^{2} = 1$, $\forall b_{k,n} \in \mathcal{M}_{l}$,
%$l=0,1,\ldots,L$,
where $\Ex$ denotes expectation
and $R_{0} = 0$  ``outage''.
For simplicity, the users $k=1,2,\ldots,K_{n}$
are assumed to be active at time instant $n$, and
transmissions are
initiated only at the beginnings of the fading blocks.
%,
%all of which have a coherence time $T_{\mathrm{b}}$.
%Recall also that the code word durations for all users
%are assumed to be less than $T_{\mathrm{b}}$.
%
With this assumption, the number of active users during a 
given fading block satisfies
$K_{n} \geq K_{n'}$ for all time instants $n<n'$.
%since the
%users may finish their transmissions in the middle of the fading block
%but are not allowed to start a new one before the next fading block.
%

If the BS has knowledge
of $\boldsymbol{\gamma}_{k}
=[\gamma_{k,1}, \ldots,\gamma_{k,N_{k}}]^{\mathsf{T}}$
for all MSs, it can use \eqref{eq:cm_capacity} to find the 
codes for the active users%
\footnote{The users, e.g., inform their
transmit payload before each fading block through a
control channel, and the BS uses
\eqref{eq:instantaneous_sinr_lmmse} (or \eqref{eq:sinr_conditioned_on_h_and_s_lmmse})
to compute  $\boldsymbol{\gamma}_{k}$.}.
This, however, may consume too much uplink bandwidth and
processing power at the BS.
Here we we consider a simplified
%Thus, if the channel is not changing too rapidly in time
%the whole set $\{\boldsymbol{\gamma}_{k}\}_{k=1}^{K}$ is known
AMC scheme where, instead of using \eqref{eq:cm_capacity},
the BS finds
\begin{equation}
  l_{k}^{*} = \max_{l = 0,1,\ldots,M}
  \{l : R_{l} \leq I^{(1)}_{\mathcal{M}_{l}}(\gamma_{k,1})\}, \quad
  k = 1,\ldots,K_{n},
  \label{eq:opt_transmission_mode}
\end{equation}
and feeds back the indexes $\{l_{k}^{*}\}_{k=1}^{K_{1}}$ to the active MSs
before the data transmission starts.
Each active MS then uses $\mathcal{C}_{k,l_{k}^{*}}$ to transmit at rate
$R_{l^{*}}$ over the CDMA channel with a vanishing probability of
error as the code word lengths and coherence time grow large since
$I^{(1)}_{\mathcal{M}_{l}}(\gamma_{k,1}) \leq
I^{(N_{k})}_{\mathcal{M}_{l}}(\boldsymbol{\gamma}_{k})$ always holds.

In theory, we would like to
have a large set of code books with finely spaced rates to
accurately adjust to the channel conditions. However,
designing arbitrary rate codes is difficult and
as the number of code books $L$ grows,
a reliable and very high rate feedback channel is needed
for code book selection.  Thus, the number of \emph{transmission modes} $L$
is in practice small.

\begin{remark}
The authors in, e.g., %\cite{Hole-Holm-Oien-2000,
\cite{Liu-Zhou-Giannakis-2004, EBWireless:Tang2007:CrossLayerModelingforQoS} propose to simulate specific AMC schemes in AWGN
channel and use parameter fitting to derive
``analytical'' error probabilities for them.
%The results from
%\cite{Liu-Zhou-Giannakis-2004} were then directly used
%in \cite{EBWireless:Tang2007:CrossLayerModelingforQoS} for
%cross layer design and system evaluation.
Such results are, however, heavily dependent on the chosen
error control codes.
We consider instead the modulation constrained capacity
\eqref{eq:cm_capacity} that gives an upper bound for all coding
schemes with the same modulation sets and code rates.
With the new code designs \cite{Richardson-Urbanke-2008},
we expect our results to
provide a close description of modern wireless systems.
%Furthermore, there is no need to simulate
%any coding schemes to get the results --- evaluation of
%\eqref{eq:cm_capacity} suffices, which can be done efficiently
%via numerical integration.
%On the other hand, modern
%graph-like codes  have
%been shown to offer almost capacity achieving
%performance under some regions SNR with complexity comparable to
%the conventional convolutional codes.
\end{remark}

\section{Delay Constrained Throughput Analysis}
    \label{Sec:DelayConstrainedThAnalysis}
%We consider an asymptomatically large user system in this work in which the states of the users decouple.
%This simplification lets us view the multiuser system as an equivalent single user system.
%We quantify the maximum traffic (delay constrained throughput) that this arbitrary user can pump into the network under the delay and burstiness constraints.
%The delay constrained throughput of the whole system would then by $\alpha$ times the single user throughput.
%
%
 %   \subsection{Block Fading Model}

% We consider a block fading model in which time is divided into equal intervals $T_{\mathrm{b}}$, which is referred to as slots,
% where $T_{\mathrm{b}}$ is the time required to transmit a block of symbols.
% The fading state is assumed to remain unchanged during the interval $T_{\mathrm{b}}$.
%We assume that the data block at the data-link layer has the same duration $T_{\mathrm{b}}$ but the number of bits per block can vary depending upon the chosen AMC.
%As a result each block has portions of the packets as the packets consists of bits.

%
%\subsection{FSMC}
%
%Adaptive modulation and coding (AMC) has gained a lot of the interest in the recent past to increase the spectral efficiency.
%Under this technique the transmitter selects the AMC modes based on the channel conditions.
%The CSI which is estimated at the receiver (BS) is fed back to the transmitters (MS's) for AMC.
%A user in deep fade is characterized by the outage mode in which it is not allowed to transmit any data.
%The AMC based wireless channel service process is modeled using FSMC.

In this section we concentrate on the decoupled single-user channel
\eqref{eq:equivalent_su_lmmse} of the first user and omit both the
user $k$ and the time $n$ indexes for notational convenience.
% Recall that a block fading model with coherence time $T_{\mathrm{b}}$ is assumed in the
% present paper.
%
The wireless channel service process in the block fading channel with AMC
is modeled using an $L$-state finite state Markov chain (FSMC) where the block length (in seconds) is denoted by $T_{\mathrm{b}}$.
Each of the transmission modes is mapped to the corresponding FSMC state.
%
%
%The channel fading process is defined through a $L$-state FSMC.
Let $\mathbf{P_c}= [ p_{l,l'} ]$ be the $L \times L$ transition matrix of the FSMC.
Assuming slow fading and a relatively small value of $T_{\mathrm{b}}$,
we have $p_{l,l'}=0$ for all $|l-l'|>1$.
The adjacent state transition probabilities are determined as~\cite{WirelessMarkov:FSMC:Wang2002finite}
\begin{equation}
\label{eq:TranProbitoi-1_itoi+1}
\left\{\begin{array}{ll}%
p_{l,l+1} \approx \frac{N \left(  \gamma_{l+1}  \right)  T_{\mathrm{b}} } { \pi_l },&\mbox{$l=1,2,\ldots,L-1,$}%
\\p_{l,l-1} \approx \frac{N \left(  \gamma_{l}  \right)  T_{\mathrm{b}} } { \pi_l },&\mbox{$l=1,2,\ldots,L$}%
\end{array}\right.
\end{equation}
%
%
%\begin{equation}
%    \label{eq:TranProbitoi+1}
%p_{l,l+1} \approx \frac{N \left(  \gamma_{l+1}  \right)  T_{\mathrm{b}} } { \pi_l }, \qquad l=1,2,\ldots,L-1,
%\end{equation}
%%
%%
%\begin{equation}
%    \label{eq:TranProbitoi-1}
%p_{l,l-1} \approx \frac{N \left(  \gamma_{l}  \right)  T_{\mathrm{b}} } { \pi_l }, \qquad l=1,2,\ldots,L,
%\end{equation}
%
%
where $N\! \left(  \gamma  \right)$ denotes the level crossing rate (LCR) at the SNR value of $\gamma$.
%It describes the average rate of up-crossings (or down-crossings) of the SNR through a certain threshold level $\gamma$.
The remaining transition probabilities are given as
\begin{equation} \label{E:TransProb_Remain}
\left\{\begin{array}{lll}%
p_{1,1} = 1 -   p_{1,2}&\mbox{}%
\\p_{L,L} = 1 -   p_{L,L-1}&\mbox{}%
\\p_{l,l} = 1 -   p_{l,l-1}  -   p_{l,l+1},&\mbox{$l=2,3,\ldots,L-1$}.%
\end{array}\right.
\end{equation}
%
%\begin{equation}
%  p_{1,1} = 1 -   p_{1,2}
%  \label{eq:TransProb_Remain1}
%\end{equation}
%%
%\begin{equation}
%  p_{L,L} = 1 -   p_{L,L-1}
%  \label{eq:TransProb_Remain2}
%\end{equation}
%%
%\begin{equation}
%  p_{l,l} = 1 -   p_{l,l-1}  -   p_{l,l+1}  , \qquad l=2,3,\ldots,L-1
%  \label{eq:TransProb_Remain3}
%\end{equation}
%
The level crossing rate reads for $\Gamma \geq 0$~\cite{WirelessMarkov:FSMC:Wang2002finite}
\begin{equation}
 N\! \left(  \Gamma  \right) =
  \sqrt{ \frac{2 \pi \Gamma} { \overline{\gamma} } } f_{\mathrm{m}}
  \exp\left( - \frac{\Gamma}{ \overline{\gamma} } \right),
  \label{eq:LevelCross}
\end{equation}
where $f_{\mathrm{m}} = v / \omega$ is the maximum Doppler frequency of the channel,
defined in terms of the vehicle speed $v$ and the wavelength $\omega$
of the transmitted signal.
The stationary probability~$\pi_l$ of the FSMC being in state $l$ is given as
\begin{equation}
    \label{eq:steadtStateProb}
\pi_l = \int_{\Gamma_l}^{\Gamma_{l+1}}  f(\gamma) d\gamma  ~=~\exp \left ( - \frac{\Gamma_l } {\overline{\gamma}}  \right ) - \exp \left ( - \frac{\Gamma_{l+1} } {\overline{\gamma}}  \right ).
\end{equation}
The PDF of the instantaneous received SNR $f(\gamma)$ 
of the single-user channel \eqref{eq:equivalent_su_lmmse}
is statistically equivalent to the large system post-detection SINR of \eqref{eq:sys_model}, and given by
%%
%\begin{figure}[t]
%	\centering		 \includegraphics[width=0.90\columnwidth]{./fig/thVSSNR_eps6_N2_d20_30_40_50_tend4000.pdf}
%	\caption{Delay constrained throughput for different SNR's}
%	\label{fig:thVSSNR_N2}
%\end{figure}
%%
%%
%
%The maximum carried traffic without the delay and burstiness constraints in $bps/Hz$ is given by
%\begin{equation}
%C_{\mathrm{lim}}=\alpha \sum_{l=1}^L \rate_l\pi_l .
%\label{eq:capacity_limit}
%\end{equation}
%
\begin{equation}
  f(\gamma) =
%  \begin{cases}
%    \displaystyle
  \frac{1}{\overline{\gamma}}
  \exp\left(-\frac{\gamma}{\overline{\gamma}} \right),
  \qquad
   \gamma \geq 0, \\
  %0 & \text{otherwise},
%  \end{cases}
  \label{eq:pdf of_SNR_SISODSCDMA}
\end{equation}
where $\overline{\gamma} = \Ex_{p_1}\{\gamma_{1,n}\} = 1/\beta_{1,n}$.
We remind the reader that the average SNR $\overline{\gamma}$
is in fact a function of time $n=1,2,\ldots,$
the noise variance at the receiver $\sigma^{2}$, and the
instantaneous user load $\alpha_{n}$ through
the fixed point equation \eqref{eq:b1_fp_lmmse}.

\begin{figure*}[t]
\setcounter{mytempeqncnt}{\value{equation}}
\setcounter{equation}{21}
\begin{IEEEeqnarray}{l}
\normalsize
\label{E:DelayBoundtwice}
d^{\varepsilon}_{\lambda} =  \inf_{\theta > 0}
\left \{ \inf_{\tau \geq 0} \left [ \tau :  \frac{1}{\theta}  \left (  \ln  \sum_{s=\tau}^{\infty} \mathsf{M}_A(\theta,s- \tau) \widehat{\mathsf{M}}_S(\theta,s) - \ln \varepsilon \right ) \leq 0 \right ]  \right \},
\end{IEEEeqnarray}
\hrulefill
\vspace*{-3ex}
\setcounter{equation}{\value{mytempeqncnt}}
\end{figure*}

\subsection{Main Result}

%We are now in a position to present the main result of this paper.
The moment generating function (MGF) of a stationary process $X(t)$ is given by $\mathsf{M}_X(\theta,t) = \Ex \left[e^{\theta X(t)}\right]$.
In the sequel, we denote
$\widehat{\mathsf{M}}_X(\theta,t) = \mathsf{M}_X(-\theta,t)$
for the parameters $\theta > 0$, $t \geq 0$.
The discrete time arrivals and service of the channel are assumed to be independent stationary random processes given by the cumulative
processes $A(0,t)$ and $S(s,t)$ respectively.
For all real $\theta$, the corresponding MGFs are
$\mathsf{M}_A(\theta,t)$ and $\widehat{\mathsf{M}}_S(\theta,t)$, respectively.
Furthermore, let $N_{\mathrm{b}}$ denote the number of information bits per an upper layer data block and $W$ the system bandwidth. 
Omitting again both the
user and the time indexes for notational convenience, 
we can then write the number of data blocks transmitted in state $l$ of the FSMC based channel service process as 
$\tilde{R}_l = \rate_l T_{\mathrm{b}} W / N_{\mathrm{b}}$.

\begin{lemma}\label{Lema:MGF_s}
The MGF of the random process $S(t)$ described by a homogeneous Markov chain with transition matrix~$\mathbf{P_c}$ and stationary state distribution vector~$\boldsymbol{\pi}$, is given for $\theta > 0$ and $t \geq 0$ by
\begin{equation}
\label{E:MS}
\widehat{\mathsf{M}}_S(\theta,t) = \boldsymbol{\pi} (\mathbf{\tilde{R}}(-\theta)\mathbf{P_c})^{t-1}\mathbf{\tilde{R}}(-\theta)\mathbf{1} ,
\end{equation}
where $\mathbf{1}$ is column vector of ones and
\begin{equation}
	 \mathbf{\tilde{R}}(\theta) = \mathrm{diag}(e^{\theta \tilde{R}_1},\dots,e^{\theta \tilde{R}_K}) .
\end{equation}
The steady state vector $ \boldsymbol{\pi}~=~\left[ \pi_1 \; \pi_2 \;\cdots\,\;\pi_L \right] $ is given in~\eqref{eq:steadtStateProb} and
$\mathbf{P_c}$ is obtained from \eqref{eq:TranProbitoi-1_itoi+1}~--~\eqref{E:TransProb_Remain}.
\end{lemma}
%    \eqref{eq:instantaneous_sinr_lmmse}~--~\eqref{eq:equivalent_su_lmmse}.
%where $\tilde{R}_l = \rate_l T_{\mathrm{b}} W$, where $W$ denotes the system spectral-bandwidth.
\begin{proof}
We refer the interested reader to \cite{NetCal:Chang00:PerGuaran:Bk}.
%for the proof.
%\qed
\end{proof}

We denote the traffic arrival rate by~$\lambda$ while $d^{\varepsilon}_{\lambda}$ represents a bound with violation probability $\varepsilon$ on the delay.
We can only provide a delay guarantee $d^{\mathrm{g}}$ if $d^{\varepsilon}_{\lambda} \leq d^{\mathrm{g}}$.
%with a violation probability~$\varepsilon$.
\begin{proposition}
The delay constrained throughput~$\lambda_{\mathrm{d}}$ of a DS-CDMA system under a delay guarantee~$d^{\mathrm{g}}$ is given as
\begin{equation}
\label{E:th}
\lambda_{\mathrm{d}} =  ~\alpha \cdot \mathrm{max}~\{\lambda ~ \vert ~ d^{\varepsilon}_{\lambda}\leq d^{\mathrm{g}} \}  ,
\end{equation}
where $d^{\varepsilon}_{\lambda}$, assuming FIFO scheduling, is given in
\eqref{E:DelayBoundtwice} at the top of the page and
$\widehat{\mathsf{M}}_S(\theta,s)$ is given by \eqref{E:MS}.
\end{proposition}
\addtocounter{equation}{1}
\begin{proof}
We refer the interested reader to \cite{NetCal:Fidler2006:EoEProbabNetCalWithMGF}.
%or the proof of \eqref{E:DelayBoundtwice}.
%\qed
\end{proof}
\begin{remark}
    \label{Rem:Dbound}
The delay bound~\eqref{E:DelayBoundtwice} is calculated using stochastic NetCal approach~\cite{NetCal:Fidler2006:EoEProbabNetCalWithMGF}.
MGF for a variety of arrival models is available in the literature~\cite{NetCal:Chang00:PerGuaran:Bk} but finding the MGF of the service process~$\widehat{\mathsf{M}}_S(\theta,t)$ is a challenging task.
We address this challenge by making use of Lemma~\ref{Lema:MGF_s} to calculate $\widehat{\mathsf{M}}_S(\theta,t).$
Having obtained $\widehat{\mathsf{M}}_S(\theta,t)$, a stochastic bound on the delay $d^{\varepsilon}_{\lambda}$ in \eqref{E:DelayBoundtwice} can be obtained using Chernoff's bound, Boole's inequality and applying the technique in \cite{NetCal:Fidler2006:EoEProbabNetCalWithMGF}.
\end{remark}

We next use the result in this section to calculate the delay constrained throughput of a multiuser DS-CDMA system.
Before closing this section, we present a definition which will be used in the next section. The maximum carried traffic without the delay constraints in $bps$ is given by
\begin{equation}
C_{\mathrm{lim}}=\alpha W \sum_{l=1}^L \rate_l\pi_l .
\label{eq:capacity_limit}
\end{equation}

%The MGF stochastic network calculus framework provides the freedom of a general separation in the characterization of arrivals and service, which allows the combination of different arrival and service models where the MGF can be calculated.
\section{Application and Numerical Results}
\label{Sec:NumericalResults}

In the numerical results, we consider
the AMC setup given in Table~\ref{tab:transmission_modes}.
The transmission modes follow the HIPERLAN/2
specification (see, e.g.,
\cite{EBWireless:Tang2007:CrossLayerModelingforQoS,
Liu-Zhou-Giannakis-2004}).
The parameter $\gamma_{l}$ in Table~\ref{tab:transmission_modes}
denotes the SNR point in decibels (dBs)
where the $l$th mode is switched on.
%
%The spectral efficiency obtained by these transmission modes in
%non-fading AWGN channel is plotted
%in Fig.~\ref{fig:speff_coded_modulation}.  We have also included there
%the capacity with optimal Gaussian inputs and variable rate coding with
%64-QAM modulation for comparison.
\begin{table}
  \centering
  \caption{Transmission modes}
\begin{tabular}[t]{llll}
  Modes & $\mathcal{M}_{l}$ & $R_{l}$ (bps/Hz) & $\gamma_{l}$ (dB)
  \\\hline
  0 &   BPSK         & 0     & -$\infty$ \\
  1 &   BPSK      & 0.5   & -2.80     \\
  2 &   QPSK      & 1     & 0.19      \\
  3 &   QPSK      & 1.5   & 3.39      \\
  4 &   16-QAM    & 2.25  & 6.20      \\
  5 &   16-QAM    & 3     & 9.30       \\
  6 &   64-QAM    & 4.5   & 14.37
\end{tabular}
  \label{tab:transmission_modes}
\end{table}
%
%\begin{figure}[t]
%	\centering		 \includegraphics[width=0.90\columnwidth]{./fig/modes}
%	\caption{ Spectral efficiency of coded modulation in
%  AWGN channel.}
%	\label{fig:speff_coded_modulation}
%\end{figure}
%
The physical layer of the system is parameterized by using a $W = 20$~MHz channel.
The base time unit is chosen to be $T_{\mathrm{b}} = 2$~ms 
and a block transmission with fixed duration $T_{\mathrm{b}}$
is assumed.
The size of the upper layer data blocks 
is fixed to $N_{\mathrm{b}} = 10\, 000$ bits, but the number of information bits 
transmitted per fading block 
varies depending upon the selected AMC scheme.
%the maximum rate given in \cite{MIMO:80211n:standard}.
The rate matrix~$\mathbf{R}(\theta)$ can be read from the Table~\ref{tab:transmission_modes}.
The infinite sum in the delay bound formula~\eqref{E:DelayBoundtwice}
is computed for the first $4000$ units of time.

The framework discussed earlier enables the derivation of delay constrained throughput of a various number of traffic sources whenever the MGF exists.
%A collection of such MGFs for sources of different statistical properties can be found in \cite{EB:Kelly96}.
In this work we use a periodic source for the discrete time model that generates arrival traffic.
Such a traffic source with period $\tau$ produces $\delta$ units of workload (data blocks)
% at times $\left\{U \tau + n \tau, n = 0,1,\ldots\right\}$ where $U$ is the starting time which is uniformly distributed  in the interval $[0,1]$.
and its MGF reads~\cite{NetCal:Fidler06:MGFfadingChannel} %of such a source is known (see e.g.~\cite{NetCal:Fidler06:MGFfadingChannel}) as
\begin{equation}
\label{E:MA_periodic}
	\mathsf{M}_{A}(\theta,t) = e^{\theta \delta \left\lfloor \frac{t}{\tau}\right\rfloor} \left(1+\left(\frac{t}{\tau} - \left\lfloor \frac{t}{\tau}\right\rfloor\right)\left(e^{\theta \delta} - 1\right)\right) ,	
\end{equation}
where $t \geq 0$ and $\theta > 0$.
Here we set $\tau ~=~ 1$ and use the number of generated data blocks $\delta$ to set the arrival rate $\lambda$.
%We fix $\tau = 10$, SNR$_{\mathrm{avg}}= 6$ dB and $\varepsilon = 10^{-2}$ unless states otherwise.

%
%
%
%\begin{figure}[t]
%	\centering		 \includegraphics[width=0.90\columnwidth]{./fig/ThroughputVSBurstiness_SYSTEM_DB100_110_120}
%	\caption{ $\mathsf{SNR}_{\mathrm{avg}}=6dB, f_m =20 Hz,  \alpha = 0.5,\varepsilon = 10^{-2}$}
%	\label{fig:thVSBurstiness_DB}
%\end{figure}
%
%
%
%

%Fig.~\ref{fig:thVSBurstiness_DB} depicts effect of burstiness on delay constrained throughout.
%The burstiness is measured in terms of number of data blocks $\delta$ generated in a time slot.
%%Each data block is composed of $1080$ bits.
%The period $\tau$ of the source allows us to set the arrival rate for the fixed burstiness $\delta$ and delay guarantee.
%It is seen that the increase in the source burstiness decreases the throughput.
%
Fig.~\ref{fig:thVSdb_N2} depicts throughput as a function of delay guarantee~$d^{\mathrm{g}}$ for different delay bound violation probabilities $\varepsilon$.
One can observe the decrease in system 
throughput as the violation probability 
and / or the delay guarantee get(s) tighter.
\begin{figure}[t]
	\centering		 \includegraphics[width=0.90\columnwidth]{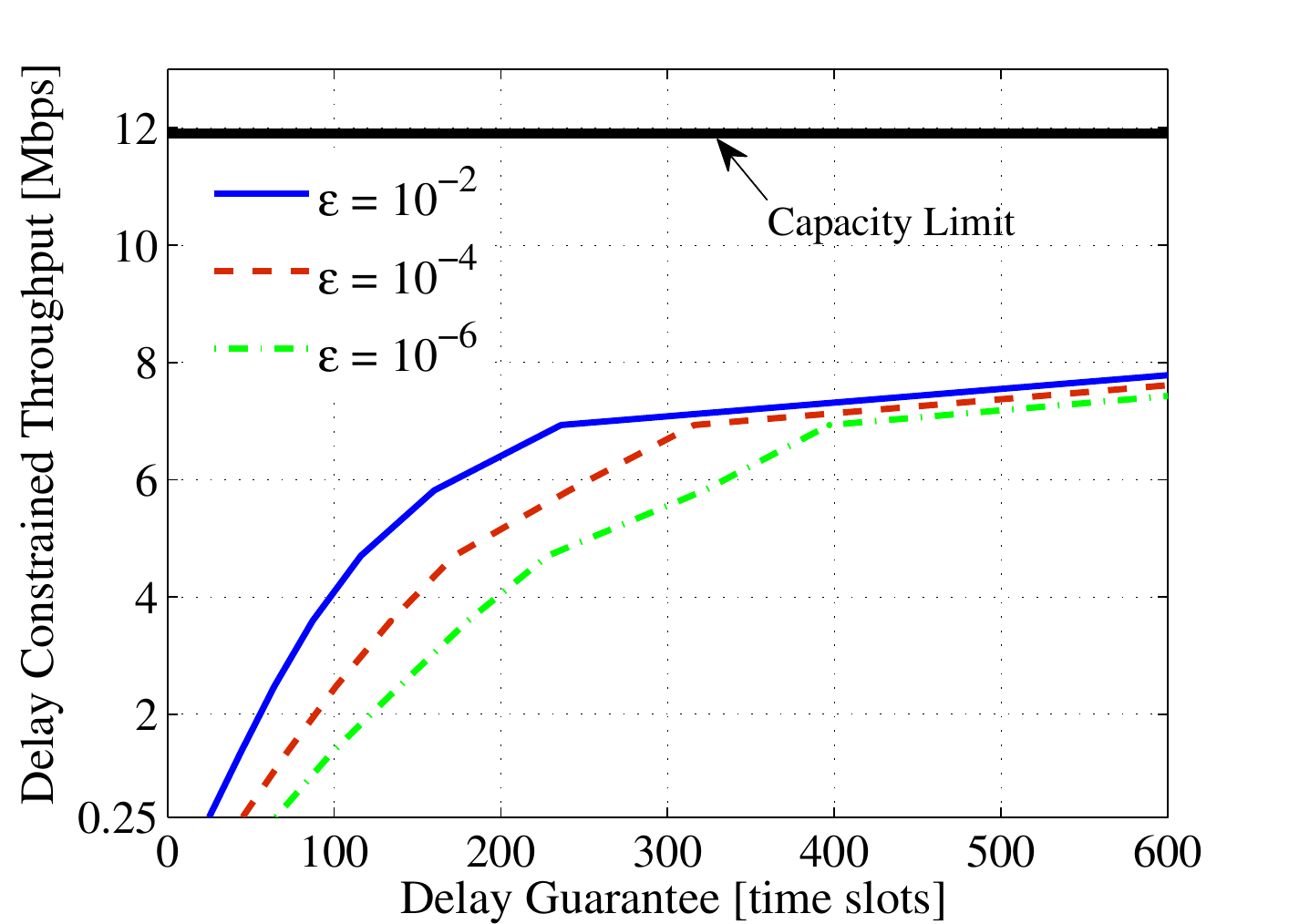}
	\caption{ $\mathsf{SNR}_{\mathrm{avg}}=6$~dB, $~f_{\mathrm{m}} =20$~Hz,  $~~\alpha = K / M = 0.5$}
	\label{fig:thVSdb_N2}
\end{figure}
\begin{figure}[t]
	\centering		 \includegraphics[width=0.90\columnwidth]{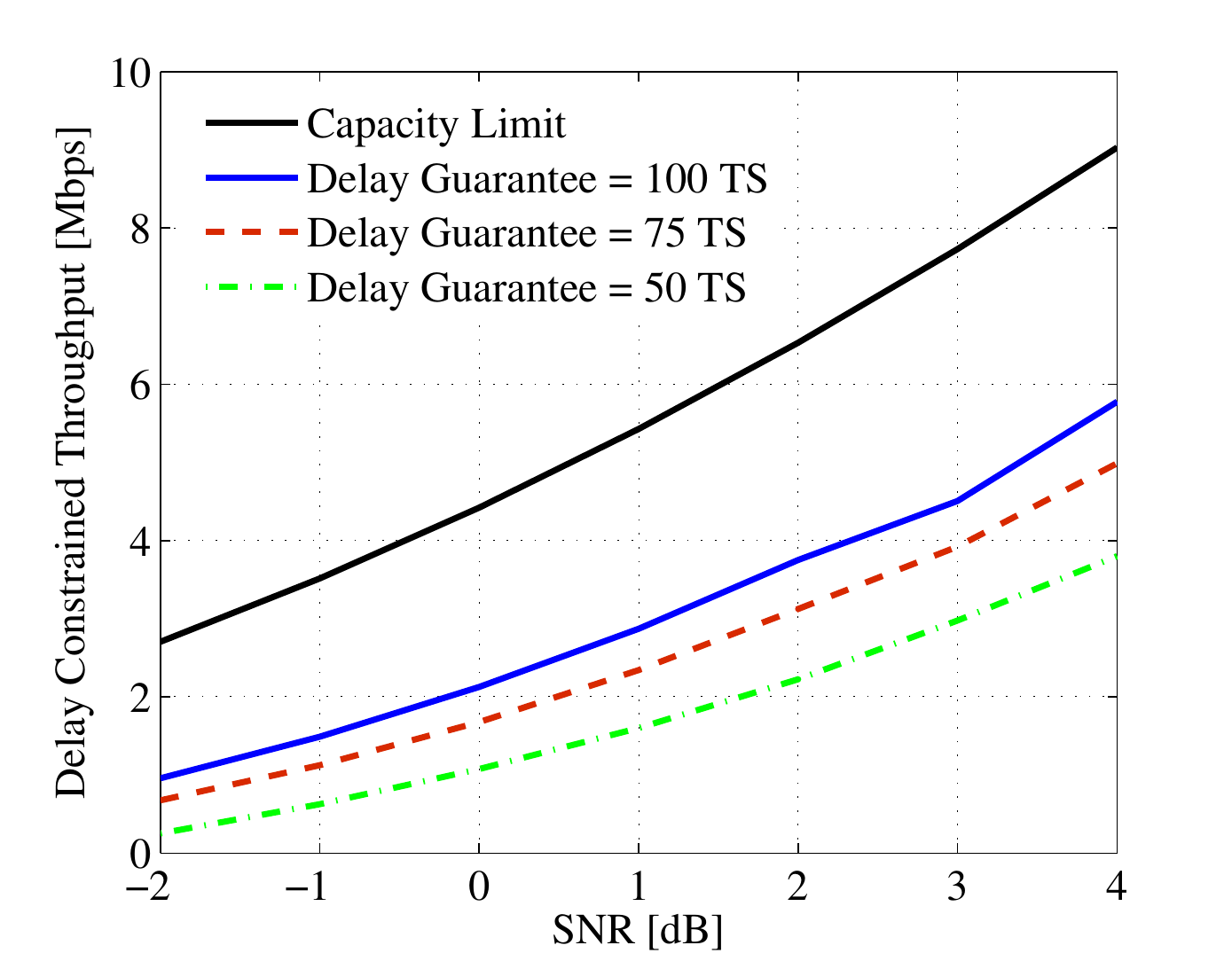}
	\caption{ $f_{\mathrm{m}} =50$~Hz,  $~\alpha = K / M = 0.5, ~\varepsilon = 10^{-2}$}%tau=10
	\label{fig:thVSSNR_DG}
\end{figure}

The impact of the average 
received SNR ($\mathsf{SNR}_{\mathrm{avg}}= 1 / \sigma^{2}$) on the 
throughput is studied in Fig.~\ref{fig:thVSSNR_DG}.  
As expected, increasing the average received SNR leads to an improvement in the throughput, as seen in the classical information theoretic results as well.
For the considered system this is due to the fact that, the higher average received SNR increases the likelihood of a better channel (note that here this effect comes from
\eqref{eq:instantaneous_sinr_lmmse}, \eqref{eq:b1_fp_lmmse},
\eqref{eq:steadtStateProb} and \eqref{eq:pdf of_SNR_SISODSCDMA})
at any given time instant and, thus, allows for higher rate code books to be used at the 
transmitters on average. This can be clearly seen by observing the steady state vectors, e.g., for $\mathsf{SNR}_{\mathrm{avg}}=-2$~dB and $\mathsf{SNR}_{\mathrm{avg}}=4$~dB which are $[0.622, 0.234, 0.127, 0.017, 0.00044, 1.43\times10^{-7}]$ and $[0.25, 0.184, 0.261, 0.203, 0.096, 0.01, 3.991\times10^{-7} ]$, respectively.
%%
%%
%\begin{figure}[t]
%	\centering
%		 \includegraphics[width=0.90\columnwidth]{./fig/ThrouputVSDB_SNR6_alpha05_fm10_20_30}
%	\caption{ $SNR_{\mathrm{avg}}=6dB, f_m =20 Hz, \varepsilon = 10^{-2}$}
%	\label{fig:thVSdb_userlaod}
%\end{figure}
%%
%
%We next study the effect of fading speed on the throughput in Fig.~\ref{fig:thVSfad}.
%It is seen that the throughput increases with the doppler frequency.
%This is because an increased doppler frequency lowers the probability that the user will have a bad channel during the interval.
%%
%\begin{figure}[t]
%	\centering
%		 \includegraphics[width=0.90\columnwidth]{./fig/ThrouputVSDB_SYSTEM_SNR6_alpha05_fm10_20_30}
%	\caption{ $SNR_{\mathrm{avg}}=6dB, ~\alpha = 0.5, ~\varepsilon = 10^{-2}$}
%	\label{fig:thVSfad}
%\end{figure}
%%
%

Finally, Fig.~\ref{fig:thVSuserlaod_DG} 
depicts the delay constrained capacity as a function of the maximum 
user load $\alpha = K / M$.  With the delay constraint, the system 
throughput behaves similarly to the information theoretic case
(see, e.g., \cite[Fig.~1]{CDMA_ImpactofFreqFlatFading:Shamai-Verdu:01}),
but the optimum user load shifts towards zero when the delay guarantee
gets tighter.  
The non-linear behavior of the throughput as a function of the 
maximum user load is due to the interplay between
\eqref{eq:b1_fp_lmmse}, \eqref{E:th} and
\eqref{eq:capacity_limit}.
We also see that the higher the user load is, the further 
we have to back-off from the capacity to meet the delay guarantee.

\begin{figure}[t]
	\centering
		 \includegraphics[width=0.90\columnwidth]{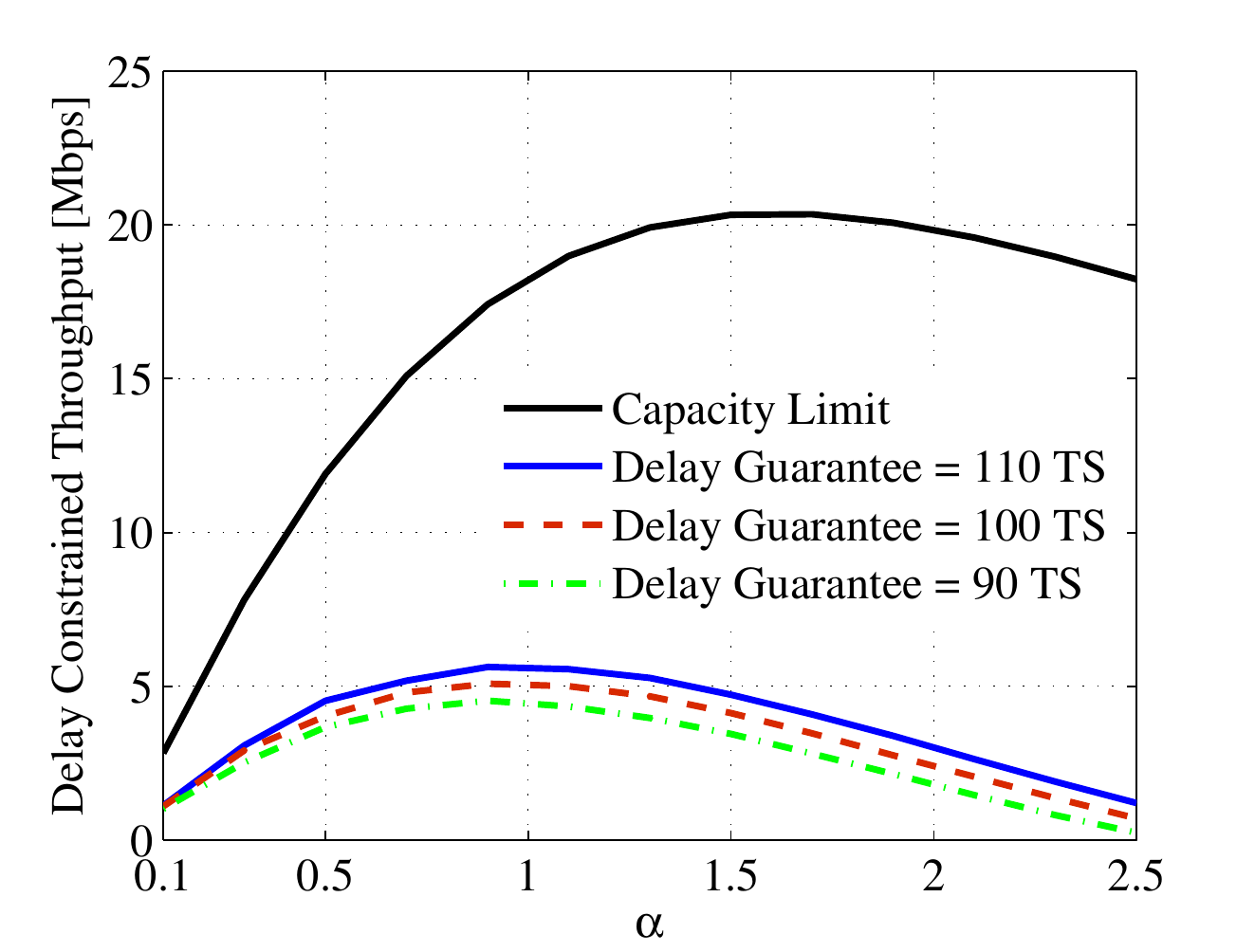}
	\caption{ $\mathsf{SNR}_{\mathrm{avg}}=6~$dB,$~f_{\mathrm{m}} = 20~$Hz,$~\varepsilon = 10^{-2}$} % tau=10, vary sigma
	\label{fig:thVSuserlaod_DG}
\end{figure}
%
%

%************  User Load fig.2  ****************
%%
%\begin{figure}[t]
%	\centering
%		 \includegraphics[width=0.90\columnwidth]{./fig/ThrouputVSDB_SNR6_fm20_alpha025_05_075}
%	\caption{ $SNR_{\mathrm{avg}}=6dB, f_m =20 Hz, \varepsilon = 10^{-2}$}
%	\label{fig:thVSdb_userlaod}
%\end{figure}
%%

%************  SNR fig.2  ****************                                                                                         %%%
%\begin{figure}[t]
%	\centering		 \includegraphics[width=0.90\columnwidth]{./fig/ThrouputVSDB_SNR6_8_10_alpha05_fm20}
%	\caption{ $f_m =20 Hz,  \alpha = 0.5,\varepsilon = 10^{-2}$}
%	\LABEL{FIG:THvssnr}
%\END{FIGURE}
%%

\section{Conclusions and Future Work}
In this paper we formulate a method to find the delay constrained throughput of a multiuser DS-CDMA system with time varying channel.  Channel memory is modeled by a finite state Markov chain.
We present numerical results where we quantify the impact of increasing the delay guarantee and the user load on the delay constrained throughput for varying signal strength.
This work finds application in performance evaluation of wireless networks where the maximum throughput for a given delay guarantee and amount of
resources is of interested.
While in this work, we only used periodic source to illustrate the results, the same methodology can be applied to any traffic source with known MGF.
%The formulation can also be extended to multihop scenarios and it is left as our future work.

\bibliographystyle{IEEEtran}
\bibliography{./XBib_Kashif_Amr_Mikko}

\begin{thebibliography}{10}
\providecommand{\url}[1]{#1}
\csname url@samestyle\endcsname
\providecommand{\newblock}{\relax}
\providecommand{\bibinfo}[2]{#2}
\providecommand{\BIBentrySTDinterwordspacing}{\spaceskip=0pt\relax}
\providecommand{\BIBentryALTinterwordstretchfactor}{4}
\providecommand{\BIBentryALTinterwordspacing}{\spaceskip=\fontdimen2\font plus
\BIBentryALTinterwordstretchfactor\fontdimen3\font minus
  \fontdimen4\font\relax}
\providecommand{\BIBforeignlanguage}[2]{{%
\expandafter\ifx\csname l@#1\endcsname\relax
\typeout{** WARNING: IEEEtran.bst: No hyphenation pattern has been}%
\typeout{** loaded for the language `#1'. Using the pattern for}%
\typeout{** the default language instead.}%
\else
\language=\csname l@#1\endcsname
\fi
#2}}
\providecommand{\BIBdecl}{\relax}
\BIBdecl

\bibitem{Verdu-1998}
S.~Verd{\'u}, \emph{Multiuser Detection}.\hskip 1em plus 0.5em minus
  0.4em\relax Cambridge, UK: Cambridge University Press, 1998.

\bibitem{Richardson-Urbanke-2008}
T.~Richardson and R.~Urbanke, \emph{Modern Coding Theory}.\hskip 1em plus 0.5em
  minus 0.4em\relax Cambridge, UK: Cambridge University Press, 2008.

\bibitem{NIT:Jan08:RethinkingIT}
{J. G. Andrews \emph{et al.}}, ``Rethinking information theory for mobile ad
  hoc networks,'' \emph{IEEE, Commun. Mag.}, vol.~46, no.~12, 2008.

\bibitem{EBWireless:DWU03}
D.~Wu and R.~Negi, ``Effective capacity: A wireless link model for support of
  quality of service,'' \emph{IEEE Trans. on Wireless Commun.}, vol.~2, no.~4,
  pp. 630--643, 2003.

\bibitem{NetCal:Fidler06:MGFfadingChannel}
M.~Fidler, ``A network calculus approach to probabilistic quality of service
  analysis of fading channels,'' in \emph{Proc. of GLOBECOM}, 2006.

\bibitem{NetCal:Jiang05:StochServGuarantServerModel}
Y.~Jiang and P.~Emstad, ``{Analysis of stochastic service guarantees in
  communication networks: A server model},'' \emph{Quality of Service--IWQoS
  2005}, pp. 233--245, 2005.

\bibitem{EBWireless:Tang2007:CrossLayerModelingforQoS}
J.~Tang and X.~Zhang, ``{Cross-layer modeling for quality of service guarantees
  over wireless links},'' \emph{IEEE Trans. on Wireless Commun.}, vol.~6,
  no.~12, pp. 4504--4512, 2007.

\bibitem{CDMA:EffectiveBW:YuF06}
F.~Yu and V.~Krishnamurthy, ``{Effective bandwidth of multimedia traffic in
  packet wireless CDMA networks with LMMSE receivers: A cross-layer
  perspective},'' \emph{IEEE Trans. on Wireless Commun.}, vol.~5, pp. 525--530,
  2006.

\bibitem{NetCal:Jiang08:SNC:Bk}
Y.~Jiang and Y.~Liu, \emph{Stochastic Network Calculus}.\hskip 1em plus 0.5em
  minus 0.4em\relax Springer, 2008.

\bibitem{NetCal:Chang00:PerGuaran:Bk}
C.-S. Chang, \emph{Performance Guarantees in Communication Networks}.\hskip 1em
  plus 0.5em minus 0.4em\relax London, UK: Springer-Verlag, 2000.

\bibitem{NetCal:Fidler2006:EoEProbabNetCalWithMGF}
M.~Fidler, ``{An end-to-end probabilistic network calculus with moment
  generating functions},'' in \emph{Quality of Service, 2006. IWQoS 2006. 14th
  IEEE International Workshop on}.\hskip 1em plus 0.5em minus 0.4em\relax IEEE,
  2006, pp. 261--270.

\bibitem{WirelessMarkov:Sadeghi08:FiniteStateMarkov:Survey}
P.~Sadeghi, R.~Kennedy, P.~Rapajic, and R.~Shams, ``{Finite-state Markov
  modeling of fading channels},'' \emph{IEEE Signal Processing Magazine},
  vol.~25, no.~5, pp. 57--80, 2008.

\bibitem{EBWireless:Li07}
C.~Li, H.~Che, and S.~Li, ``{A wireless channel capacity model for quality of
  service},'' \emph{IEEE Trans. on Wireless Commun.}, vol.~6, no.~1, 2007.

\bibitem{CDMA:LinearMultiuserReceivers:Tse-Hanly-99}
D.~N.~C. Tse and S.~V. Hanly, ``Linear multiuser receivers: Effective
  interference, effective bandwidth and user capacity,'' \emph{IEEE Trans.
  Inform. Theory}, vol.~45, no.~2, pp. 641--657, March 1999.

\bibitem{Fabregas-Caire-2006}
A.~{Guill{\'e}n i F{\`a}bregas} and G.~Caire, ``Coded modulation in the
  block-fading channel: Coding theorems and code construction,'' \emph{IEEE
  Trans. Inform. Theory}, vol.~52, no.~1, pp. 91--114, Mar. 2006.

\bibitem{CDMA_ImpactofFreqFlatFading:Shamai-Verdu:01}
S.~Shamai and S.~Verd{\'u}, ``The impact of frequency-flat fading on the
  spectral efficiency of {CDMA},'' \emph{IEEE Trans. Inform. Theory}, vol.~47,
  no.~4, pp. 1302--1327, May 2001.

\bibitem{CDMA:OutputMAI:Zhang-Chong-Tse-01}
J.~Zhang, E.~K.~P. Chong, and D.~N.~C. Tse, ``Output {MAI} distributions of
  linear {MMSE} multiuser receivers in {DS-CDMA} systems,'' \emph{IEEE Trans.
  Inform. Theory}, vol.~47, no.~3, pp. 1128--1144, March 2001.

\bibitem{Liu-Zhou-Giannakis-2004}
Q.~Liu, S.~Zhou, and G.~B. Giannakis, ``Cross-layer combining of adaptive
  modulation and coding with truncated {ARQ} over wireless links,'' \emph{IEEE
  Trans. on Wireless Commun.}, vol.~3, pp. 1746--1755, 2004.

\bibitem{WirelessMarkov:FSMC:Wang2002finite}
H.~Wang and N.~Moayeri, ``{Finite-state {M}arkov channel-a useful model for
  radio communication channels},'' \emph{IEEE Trans. on Vehicular Tech.},
  vol.~44, no.~1, pp. 163--171, 2002.

\end{thebibliography}

\end{document}